\documentclass[12pt]{amsart}
\usepackage{amsmath,amssymb,amsmath,amsthm,amsopn}
\usepackage{units}
\usepackage{xfrac}
\usepackage{fixltx2e}
\usepackage[utf8]{inputenc}
\usepackage{lmodern}
\usepackage[T1]{fontenc}
\usepackage{graphicx}
\usepackage{enumitem}
\usepackage{tikz}
\usepackage{nameref}
\usetikzlibrary{patterns}

\usepackage[letterpaper]{geometry}
\usepackage[unicode=true,
bookmarks=true,bookmarksnumbered=false,bookmarksopen=false,
 breaklinks=false,pdfborder={0 0 1},backref=false,colorlinks=false]
 {hyperref}
\hypersetup{
 pdfauthor={Jeffrey Schenker}}

\numberwithin{equation}{section}
\numberwithin{figure}{section}
\newtheorem{thm}{\protect\theoremname}
\newtheorem{lem}{Lemma}

\newtheorem{cor}{Corollary}

\theoremstyle{definition}
\newtheorem{ass}{Assumption}
  \theoremstyle{remark}
\newtheorem*{rem*}{\protect\remarkname}
\newtheorem*{rems*}{Remarks}
 \providecommand{\remarkname}{Remark}
\providecommand{\theoremname}{Theorem}

\newcommand{\bb}[1]{\mathbb{#1}}
\newcommand{\mc}[1]{\mathcal{#1}}
\newcommand{\ipc}[2]{\left \langle #1 , \ #2 \right \rangle }
\newcommand{\Ev}[1]{\E \left( #1 \right)}  %% produces \E( # )
\newcommand{\norm}[1]{\left\Vert#1\right\Vert}
\newcommand{\abs}[1]{\left\vert#1\right\vert}

\newcommand{\setb}[2]{\left \{ #1 \ \middle | \ #2 \right \} }

\renewcommand{\vec}[1]{\mathbf{#1}}

\def\e{\mathrm e}

\def\im{\mathrm i}

\def\1{{\mathsf 1}}
\def\di{\mathrm d}
\def\Pr{\operatorname{Prob}} %%% appears in many equations  Prob

\def\dist{\operatorname{dist}}   %%%  distance
\def\tr{\operatorname{tr}}    %%% Trace
\def\Re{\operatorname{Re}}
\def\Im{\operatorname{Im}}
\def\Z{\mathbb Z}

\def\R{\mathbb R}
\def\C{\mathbb C}
\def\E{\mathbb E}

\def\Hi{\mathcal H}

\begin{document}
\title{Diffusive scaling for all moments of the Markov Anderson model}

\dedicatory{Dedicated to Professor Leonid Pastur on the occasion of his 75\textsuperscript{th} birthday.}

\author[Musselman]{Clark Musselman}
\address{Clark Musselman \\ Department of Mathematics \\ Bard College at Simon's Rock }
\email{bmusselman@simons-rock.edu}
\author[Schenker]{Jeffrey Schenker}
\address{Jeffrey Schenker \\ Department of Mathematics \\ Michigan State University }
\curraddr{School of Mathematics \\ Institute for Advanced Study}
\email{jeffrey@math.msu.edu}
\thanks{Some of the results presented here appeared in the Ph.D. thesis of the first author \cite{bcmphd}.}
\thanks{Both authors were supported in part by NSF Award DMS-08446325.  The second author was also supported also by The Fund For Math.}
\date{}
\maketitle

\begin{abstract} We consider a tight-binding Schr\"odinger equation with time dependent diagonal noise, given as a function of a Markov process.  This model was considered previously by Kang and Schenker \cite{Kang2009b}, who proved that the wave propagates diffusively.  We revisit the proof of diffusion so as to obtain a uniform bound on exponential moments of the wave amplitude and a central limit theorem that implies, in particular, diffusive scaling for all position moments of the mean wave amplitude.  \end{abstract}

\section{Introduction}
It is generally expected that the amplitude of a  wave propagating in a weakly disordered background evolves \emph{diffusively}, i.e., that the wave amplitude obeys an effective parabolic equation over sufficiently long space and time scales, at least in dimension $d \ge 3$.  This belief is suggested by picturing wave propagation as a multiple scattering process. Scattering off the disordered background results in random phases and the build up of these phases over time  eventually leads to decoherence.  Decoherent propagation of the wave may be understood as a classical superposition of reflections from random obstacles.  As long as recurrence does not dominate, the central limit theorem suggests a diffusive evolution for the amplitude in the long run.   So far, it has not been possible to turn this heuristic argument into mathematical analysis without restricting the time scale over which the wave evolution is followed, as in \cite{Erdos2007b,Erdos2008d}. 
   
One major obstacle to proving diffusion is a lack of control over recurrence:  the wave packet may return often to regions visited previously, denying us the independence needed to carry out the central limit argument. Thus, one may expect to produce a model in which diffusion occurs by eliminating or reducing recurrence. This is the basis for the \emph{Markov-Anderson model} considered in \cite{Pillet1985,Tcheremchantsev1997,Tcheremchantsev1998,Kang2009b}.  In particular, %\textcite{Kang2009b} 
Kang and Schenker \cite{Kang2009b} considered the equation
\begin{equation}\label{eq:MSE}\im \partial_{t} \psi_{t}(x) \ = \ K \psi_{t}(x) + v_x(\omega_t) \psi_{t}(x)
\end{equation}
on $\ell^{2}(\Z^{d})$ where 
\begin{enumerate}[topsep=1pt]
\item $K$ is a self-adjoint translation invariant hopping operator,
\begin{equation}\label{eq:h} K \psi(x)=\sum_{y} h(x-y) \psi(y),
\end{equation}
with $h$ such that $\sum_{y} y^2 \abs{h(y)} <\infty$ and satisfying a non-degeneracy condition explicated in \cite{Kang2009b},
\item $(\omega_t)_{t\ge 0}$ is a Markov process with unique invariant probability measure $\mu$ on a space $\Omega$, and
\item $v_x:\Omega \rightarrow \R$ is a bounded ergodic random field.
\end{enumerate}
(These conditions are explained in more detail below.)  Roughly speaking, the  main  result of \cite{Kang2009b} is as follows.  If the Markov process $(\omega_t)_{t\ge 0}$  has a ``spectral gap,'' which is  to say
$$ \abs{\int_{\Omega}\Ev{f(\omega_{t})\middle | \omega_0 =\omega }d\mu(\omega) - \int_{\Omega}f(\omega) \di \mu(\omega)} \le \e^{- t/T}$$
for some $T > 0$, then any solution to eq.~ \eqref{eq:MSE} with $\norm{\psi_0}_{\ell^2} =1$ satisfies 
\begin{equation}\label{eq:diffusion}\sum_{x} e^{i \frac{1}{\sqrt{\tau}}\vec{k}\cdot x} \Ev{ \abs{\psi_{\tau t}(x) }^2} \ \xrightarrow[]{\tau \rightarrow \infty} \ \e^{-\frac{1}{2} t\vec{k} \cdot \vec{D} \vec{k} }
\end{equation}
for each $\vec{k}\in \R^d$ and $t>0$.   Here, $\vec{D}$ is a positive definite matrix, called the \emph{diffusion matrix}, that depends on $T$, $v_x$ and the process $(\omega_t)_{t\ge 0}$.  Furthermore, it was shown in \cite{Kang2009b} that if $\sum_x |x|^2 \abs{\psi_0(x)}^2 < \infty$ then 
\begin{equation}\label{eq:diffusivescaling}
\frac{1}{t} \sum_{x} |x|^2 \Ev{\abs{\psi_{t}(x)}^2} \xrightarrow[]{t\rightarrow \infty} \tr \vec{D}.
\end{equation}
These results hold in any dimension $d$, essentially because the Markovian time dependence of $v_x(\omega_t)$ eliminates recurrence effects that might otherwise dominate in low dimension.

We refer to eq. \eqref{eq:diffusion} as \emph{quantum diffusion}.  It indicates how the dynamics of $\Ev{|\psi_t(x)|^2}$ are homogenized over large space and time scales: over large spatial scales of order $\sqrt{\tau}$ and long time scales of order $\tau$, the mean square amplitude $\Ev{\abs{\psi_{ t}(x)}^2}$ is effectively described by the fundamental solution to the diffusion equation
$$\partial_t u(x,t) = \frac{1}{2} \sum_{i,j} D_{i,j} \partial_{x_i}\partial_{x_j} u(x,t).$$ 
By contrast, eq.~ \eqref{eq:diffusivescaling} states only that the mean square position of the wave packet scales linearly as $t\rightarrow \infty$.  For this reason we refer to eq.~\eqref{eq:diffusivescaling} as \emph{diffusive scaling}.  Formally, this seems to follow from eq.~\eqref{eq:diffusion} by taking two derivatives with respect to $\vec{k}$.  However, this formal differentiation is not clearly justified and in \cite{Kang2009b} the two statements were proved separately. (By L\'{e}vy's continuity theorem \cite[Theorem 3.3.6]{Durrett}, eq.~\eqref{eq:diffusivescaling} does follow from eq.~\eqref{eq:diffusion} and a uniform bound on the left hand side.  However, proving the convergence directly seems to be as simple in this case as obtaining the needed upper bound.)

In the present work, we review the proof of diffusion given in \cite{Kang2009b} under somewhat stronger conditions on the hopping operator $K$, so as to obtain diffusive scaling for all position moments,
\begin{equation}\label{eq:diffusivescalingconjecture}\sum_{x} |x|^p \Ev{\abs{\psi_t(x)}^2} \ \sim \ t^{\nicefrac{p}{2}}
\end{equation}
for each $p>0$.   Scaling as in eq.~\eqref{eq:diffusivescalingconjecture} does not follow directly from eq.~\eqref{eq:diffusivescaling}; in principle the evolution could be ``multi-fractal,'' requiring different scaling exponents for different moments.  However, this does not seem consistent with eq.~\eqref{eq:diffusion} since for even $p$ eq.~\eqref{eq:diffusivescalingconjecture} is suggested by $p$ formal differentiations with respect to $\vec{k}$.   Below, we  show that eq.~\eqref{eq:diffusion} may be extended to complex values of $\vec{k}$ so that the differentiations required to obtain scaling for the moments are justified by complex analysis. (A similar argument was used recently in the context of a discrete time quantum walks \cite{Hamza2012}. An alternative proof of the scaling eq.~\eqref{eq:diffusivescalingconjecture} could proceed by L\'{e}vy's continuity theorem and the uniform bound on exponential moments that is implied by our analytic continuation argument.)

\section{Formulation of the results}
\subsection{Simple models}  The general conditions we require of the Markov process are somewhat complicated to state.  For this reason, it may be useful first to give a few examples of models for which our results stated below are valid.  An impatient reader may wish to read these examples and then skip to \S\ref{sec:results} where our main results are stated.

In formulating these simple examples, there is no loss in taking the hopping operator to be simply the discrete Laplacian
$$K \psi(x) = \sum_{|y-x|=1} \psi(y).$$
Below, we will allow for a more general hopping operator of the form eq.\ \eqref{eq:h}.  For the potential and the Markov process we describe three examples, each of which satisfies all of the hypotheses we lay out in the next section. \begin{enumerate}[leftmargin=.275in]
\item  Let the potential $v_x(\omega_t)$ be $\pm 1$ at each site $x$ with flips from $1$ to $-1$ occurring at rate $\nicefrac{(1-p)}{T}$ and from $-1$ to $1$ at rate $\nicefrac{p}{T}$, independently for each site $x$.\\
More formally, this amounts to taking $\Omega = \{ -1,+1\}^{\Z^d}$ with  $v_x(\omega)=\omega(x)$ the coordinate maps.  The probability measure $\mu$ is the product of identically Bernoulli measures, $\mu  =  \prod_{x} \nu_x$
where $\nu_x =\nu$ is the fixed measure on $\{-1,1\}$ with
$$\nu( \{1\})  = p \quad \text{and} \quad \nu(\{-1\} ) = (1-p).$$
Finally, to each site $x\in \Z^d$ we associate a Poisson process $0<t_1(x)<t_2(x)<\cdots$ of intensity $\nicefrac{1}{T}$ such that processes associated to distinct sites are independent.  The value of $\omega_t(x)$ is constant except at the times $t_j(x)$, $j=1,\ldots$ but at time $t=t_j(x)$ it is ``resampled'' according to the measure $\nu$.  That is,
\begin{multline*}\bb{P}\left (\omega_{t+dt}(x) = 1 \middle | \omega_t(x) =1 \right )= 1 - \frac{1-p}{T} dt , \quad \bb{P}\left (\omega_{t+dt}(x) = -1 \middle | \omega_t(x) =1 \right ) = \frac{1-p}{T} dt,  \\
\bb{P}\left (\omega_{t+dt}(x) = 1 \middle | \omega_t(x) =-1 \right )= \frac{p}{T} dt  \quad \text{and} \quad \bb{P}\left (\omega_{t+dt}(x) = -1 \middle | \omega_t(x) =-1 \right ) = 1-\frac{p}{T} dt
\end{multline*}
up to terms of order $(dt)^2$.

 \item Let $\Omega$, $\mu$ and $v_x$ be as in the previous example.  However, now take only one Poisson process for all sites $0 <t_1<\cdots$ and at time $t=t_1$ ``resample'' the entire random field $\{\omega(x)\}_{x\in \Z^d}$ according to the measure $\mu$.  That is,
$$
 \bb{P}\left (\omega_{t+dt} \in A \middle | \omega_t \right ) \ = \ \begin{cases}
 1 - \frac{\mu(A^c)}{T} dt & \text{ if } \omega_t \in A \\
 \frac{\mu(A)}{T} dt & \text{ if } \omega_t \not \in A
 \end{cases}
 $$
up to terms of order $(dt)^2$

\item Let $\Omega = S_1^{\Z^d}$ where $S_1$ is the circle, parameterized as $[0,2\pi]$ with $0$ and $2\pi$ identified.  Let $v_x(\omega)= \cos(\omega(x))$ and let $\omega_t(x)=\omega_0(x) + \frac{1}{\sqrt{T}} b_t(x)$ where $b_t(x)$ are independent Brownian motions, one for each site $x\in \Z^d$.  Then the measure $\mu =\prod_x \nu_x$ with $\nu_x=\nu$ normalized Lebesgue measure is invariant.
\end{enumerate}

\subsection{Assumptions} We start with the hopping operator $K$.  We require three assumptions: self-adjointness (to guarantee unitarity of the evolution), an exponential bound (to allow for an analytic continuation argument) and non-degeneracy (to assure that diffusion is non-zero in all directions).
\begin{ass}\label{ass:K} The hopping kernel $h$ (see eq.\ \eqref{eq:h}) 
satisfies 
\begin{enumerate}
\item \emph{Self-adjointness}: $h(\zeta) = h(-\zeta)^*$
\item \emph{Exponential boundedness}:
\begin{equation}
\label{eq:hassume} 
\sum_y \e^{m|y|} h(y) \ < \ \infty
\end{equation}
for some $m>0$.
\item \emph{Non-degeneracy}: For every non-zero $\vec{k}\in \R^d$ there is $y\in \Z^d$ with $\vec{k}\cdot y \neq 0$ and $h(y)\neq 0$. 
\end{enumerate}
\end{ass}

The shifting random potential is composed of two ingredients: a family of random variables $v_x$ over a probability space $(\Omega,\mu)$ and a Markov process $(\omega_t)_{t\ge 0}$ defined on the same space. Regarding this probability space we make the following
\begin{ass} The space $\Omega$ is a topological space and $\mu$ is a Borel measure. Furthermore there are $\mu$-measure preserving maps $\sigma_x:\Omega \rightarrow \Omega$, $x\in \Z^d$, such that $\sigma_0$ is the identity map and $\sigma_x \circ \sigma_y = \sigma_{x+y}$ for each $x,y\in \Z^d$.
\end{ass}

The prototypical example of $\Omega$ is, as in each of the examples from the prior section, the model space for an infinite collection of independent identically distributed random variables.  That is, $X^{\Z^d}$ with $X$ a compact set endowed with a probability measure $\nu$,  the measure $\mu=\prod_x \nu$ and the maps $\sigma_x$ are the coordinate shifts.  However, independence and product measures are not required for our analysis \textemdash \ in particular we can allow for correlated measures (like Gibbs measures) on $X^{\Z^d}$ provided eq.~\eqref{eq:vnondeg} below holds.  We continue to refer to the maps $\sigma_x$ as \emph{shifts} below, even if $\Omega$ is not a product space.  

Let us now turn to the Markov process $(\omega_t)_{t\ge 0}$.
\begin{ass}
For each $a\in \Omega$, there is a probability measure  $\mathbb{P}_a$  on the $\sigma$-algebra generated by Borel cylinder subsets of the path space $\mc{P}(\Omega)=\Omega^{[0,\infty)}$.  Furthermore the collection of these measures has the following properties.
\begin{enumerate}
\item \emph{Right continuity of paths}: For each $a\in \Omega$, with $\mathbb{P}_a$ probability one, every path $t\mapsto \omega_t$ is right continuous and has initial value $\omega_0=a$.
\item \emph{Shift invariance in distribution}: For each $a\in \Omega$ and $x\in \Z^d$, $\bb{P}_{\sigma_x a}  =  \bb{P}_a \circ \mc{S}_x^{-1}$,
where $\mc{S}_x(\omega_t) = \sigma_x\omega_t$ is the shift $\sigma_x$ lifted to path space $\mc{P}(\Omega)$
\item \emph{Markov property}:  For each $s>0$, let $\mc{T}_s:\mc{P}(\Omega) \rightarrow \mc{P}(\Omega)$ be the time-shift operator $[\mc{T}_s\omega]_t = \omega_{t+s}$. Then, for any measurable $\mathcal{A} \subset\mc{P}(\Omega)$ and any $a\in \Omega$, $\bb{P}_a(\mc{T}_t^{-1}(\mathcal{A}))  =   \bb{E}_{a} \left ( \bb{P}_{\omega_t}(\mc{A}) \right )  $,
where $\bb{E}_{a}$ denotes the average  with respect to $\bb{P}_a$.
\item \emph{Invariance of $\mu$}: For any Borel measurable $A\subset \Omega$ and each $t>0$,
$$\int_\Omega \bb{P}_a (\omega_t\in A) \di \mu(a) \ = \ \mu(A).$$
\end{enumerate}
We will use $\Ev{\cdot}$ to denote the joint average with respect to $\bb{E}_a(\cdot)$ and $\di \mu (a)$:
$$\Ev{\cdot} \ = \ \int_\Omega \bb{E}_a(\cdot) \di \mu(a).$$ 
So invariance of the measure $\mu$ can be put succinctly as $\Ev{f(\omega_t)} = \Ev{f(\omega_0)}$ for any integrable $f:\Omega \rightarrow \R$.
\end{ass}

It is standard that the Markov property implies that
$$ S_tf(a) \ := \ \bb{E}_a\left ( f(\omega_t) \right )$$
defines  a contraction semi-group on each $L^p(\di \mu)$ space for $1\le p \le \infty$. Right continuity of paths in turn implies that this semi-group is strongly continuous.   The adjoint semi-group $S_t^\dagger$, defined by
\begin{equation}\label{eq:Stdag}
\int_\Omega g(a) S_t^\dagger f(a) \ = \ \Ev{g(\omega_t)f(\omega_0)},
\end{equation}
is also a strongly continuous contraction semi-group.\footnote{Formally, eq.~\eqref{eq:Stdag} can be read as  $S_t^\dagger f(a) = \bb{E} \left ( f(\omega_0) \middle | \omega_t=a \right ).$
This is an inkling of the more general statement, which we need not formulate precisely here, that $S_t^\dagger$ is the contraction semi-group associated to the Markov process $(\omega_t)_{t\ge 0}$ run backwards in time.}  By the Lumer-Phillips theorem,\footnote{The Lumer-Phillips Theorem \cite[Theorem 3.1]{Lumer1961} characterizes the generators of strongly continuous contraction semi-groups on a Banach space.  The Hilbert space case which we use here is in fact an earlier result due to Phillips \cite[Theorem 1.1.3]{Phillips1959}.} 
the generator $B$ defined by
\begin{equation}\label{eq:Bdefn}
Bf \ = \ \lim_{t \downarrow 0} \frac{1}{t} \left (f - S_t^\dagger f \right ),
\end{equation}
on the domain $\mathcal{D}(B)$ such that the limit on the right hand side exists in the $L^2$-norm, is a \emph{maximally accretive operator}.\footnote{
We use the word \emph{generator} to indicate that formally $S_t^\dagger = \e^{-t B}$ \textemdash \ note the negative sign in the exponent. 
A closed densely defined operator $A$ on a Hilbert space is \emph{accretive} if $\Re \ipc{f}{Af} \ge 0$ for all $f\in \mathcal{D}(A)$.  It is \emph{maximally accretive} if it has no proper closed accretive extension; equivalently $A^\dagger$ is also accretive. See \cite[\S V.3.10]{Kato1995} and \cite{Phillips1959}.}
The generator of $S_t$ is the adjoint $B^\dagger$ of $B$, which is also maximally accretive.

The invariance of $\mu$ implies that $S_t^\dagger 1 = S_t 1 =1$, where $1$ denotes the function identically one on $\Omega$. Thus $1\in \mathcal{D}(B)\cap \mathcal{D}(B^\dagger)$ and
$$ B1 = B^\dagger 1 = 0.$$
It follows that the mean-zero space 
$$L^2_0(\Omega) \ := \ \setb{f\in L^2(\Omega)}{ \int_\Omega f \di \mu = 0},$$
which is the orthogonal complement of $1$, is invariant under each of the semi-groups $S_t$ and $S_t^\dagger$. Thus
$$\mathcal{D}_0(B) \ :=\ \mathcal{D}(B) \cap L^2_0(\Omega) \quad \text{ and } \quad \mathcal{D}_0(B^\dagger) \ := \ \mathcal{D}(B^\dagger) \cap L^2_0(\Omega)$$
are each dense in $L^2_0(\Omega)$. We require strict accretivity for $B$ and $B^\dagger$ on $L^2_0(\Omega)$.  For technical reasons, related to the controlling perturbations of the semi-group $\e^{-tB}$, we also assume that $B$ is \emph{sectorial}.
\begin{ass}[\emph{Gap Condition and Sectoriality of $B$}]  There are $T >0$ and $b,q\in \R$  such that 
\begin{equation}\label{eq:gap} \Re \ipc{f}{Bf} \ \ge \ \frac{1}{T} \norm{f}_{2}^2
\end{equation}
and 
\begin{equation}\label{eq:sector}
\abs{\Im \ipc{f}{Bf}} \ \le \ q \abs{\Re \ipc{f}{Bf} + b }
\end{equation}
for all $f\in \mathcal{D}_0(B)$.  Here $\ipc{f}{g} = \int_\Omega \overline{f} g \di \mu$ denotes the inner product on $L^2(\Omega)$.
\end{ass}

The parameter $T$ represents the characteristic time scale for the process to decorrelate.   For example, it follows from the gap condition eq.~\eqref{eq:gap} that
$$\Ev{f(\omega_t) g(\omega_s)} \ = \ \int_\Omega f \di \mu \int_\Omega g\di \mu \ + \ \mathcal{O} \left ( \e^{-\nicefrac{|t-s|}{T}} \right ), \quad |t-s| \rightarrow \infty.$$
Another consequence of the gap condition eq.\ \eqref{eq:gap} is that the generator $B$ is invertible on $L^2_0(\Omega)$.  We will abuse  notation and use $B^{-1}$ to denote the inverse of $B \mid_{L^2(\Omega)}$.  So 
$$ B^{-1} \ := \ \lim_{\epsilon \rightarrow 0} B (B + \epsilon I)^{-2}$$
where $I$ denotes the identity map.

Finally, we require the following for the functions $v_x$ giving rise to the potential in eq.~\eqref{eq:MSE}.
\begin{ass}[\emph{Translation covariance and non-degeneracy of the potential}]
The functions $v_x$ are given by $v_x = v_0 \circ \sigma_x$ where $v_0 \in L^\infty(\mu)$. Furthermore, there is $\chi >0$  such that
\begin{equation}\label{eq:vnondeg}\norm{B^{-1} (v_x - v_0 )}_{L^2(\Omega)} \ \ge \ \chi , \quad x\neq 0, \ x \in \Z^d.
\end{equation}
\end{ass}

Since the Markov process is distributionally translation invariant, the generator commutes with shifts: $B(f\circ \sigma_x) = (Bf)\circ\sigma_x$. Thus eq.\ \eqref{eq:vnondeg} implies
$$\inf_{x\neq y} \norm{B^{-1} (v_x - v_y)}_{L^2(\Omega)} \ \ge \ \chi.$$
The main point of this condition is that the potentials at different sites are \emph{different} in a uniform quantitative way.\footnote{This condition is trivial to verify if, as in the first and third simple examples from the prior section, the potentials $v_x(\omega)$ are independent and undergo independent Markov processes.  Indeed, in that case
$$\norm{B^{-1} (v_x -v_0)}^2 \ = \ \norm{B^{-1} (v_x - \Ev{v_x})}^2 + \norm{B^{-1} (v_0 - \Ev{v_0})}^2 \ = \ 2 \norm{B^{-1} (v_0 - \Ev{v_0})}^2$$
is independent of $x\neq 0$ by translation invariance and since the cross terms
$$\ipc{B^{-1} (v_x -\Ev{v_x})}{B^{-1} (v_0 - \Ev{v_0})} \quad \text{and} \quad \ipc{B^{-1} (v_0 - \Ev{v_0})}{B^{-1} (v_x -\Ev{v_x})}$$
vanish by independence.}

\subsection{Results}\label{sec:results} Given $\psi_0\in \ell^2(\Z^d)$ we define 
\begin{equation}\label{eq:CF}
M_t(\vec{z};\psi_0) = \sum_{x\in \Z^d} \e^{\im \vec{z}\cdot x} \Ev{\abs{\psi_t(x)}^2}
\end{equation}
where $\psi_t$ is the solution to eq.~\eqref{eq:MSE} with initial value $\psi_0$.   For $\vec{z} \in \R^d$, this is the \emph{characteristic function} (CF) of the $\Z^d$-valued random variable $X$ with distribution  $\Pr(X=x)= \mathbb{E}({\abs{\psi_t(x)}^2})$.    Note that formally we may recover position moments from $M_t(\vec{z};\psi_0)$ by evaluating various derivatives with respect to $\vec{z}$ at $\vec{z}=0$.  For example
$$\sum_{x\in \Z^d} \abs{x}^2 \Ev{\abs{\psi_t(x)}^2} \ = \ - \left . \Delta M_t(\vec{z};\psi_0) \right |_{\vec{z}=0}$$
where $\Delta = \nicefrac{\partial^2}{\partial z_1^2} + \cdots + \nicefrac{\partial^2}{\partial z_d^2}$ denotes the Laplacian.

For $\vec{z}\in \R^d$, $M_t(\vec{z},\psi_0)$ is the Fourier series transform of the $\ell^1$ function $\mathbb{E} ( \abs{\psi_t(x)}^2)$.  In particular, it is a continuous function of $\vec{z}$ bounded by $\norm{\psi_0}_{\ell^2}^2$. However, for our main result, we will be concerned with complex $\vec{z}$.  For this purpose, we will need to restrict our attention to initial vectors $\psi_0$ satisfying an exponential bound 
\begin{equation}\label{eq:Amu}  \ \sup_{x} \e^{\mu|x|}\abs{\psi_0(x)} < \infty\end{equation}
for some $\mu >0$ \textemdash \ without loss of generality we may take $\mu <m$ with $m$ as in Assumption \ref{ass:K}.  For such $\psi_0$, we will prove that $M_t(\vec{z},\psi_0)$ defines an analytic function of $\vec{z}$ in the rectangular neighborhood 
\begin{equation}\label{eq:neighborhood}R_\mu \ = \ \setb{\vec{z}\in \C^d}{ \abs{\mathrm{Im} \ z_i} < \mu, \ i=1,\ldots,d}.
\end{equation}
of $\R^d$ in $\C^d$.
\begin{lem}[Analyticity of the CF]\label{lem:CF}
If the exponential bound eq.~ \eqref{eq:Amu} holds with $\mu <m$, where  $m$ is as in Assumption \ref{ass:K}, then for each $t \ge0$ the right hand side of eq.\ \eqref{eq:CF} is absolutely convergent and defines an analytic function of $\vec{z}$ for $\vec{z}\in R_\mu$.
\end{lem}

Our main result is the following theorem on the limit of the diffusively rescaled characteristic function.
\begin{thm}[Diffusive limit for the CF]\label{thm:CF} There is a positive definite matrix  $\vec{D}  =  (D_{i,j})_{i,j=1}^d$ such that
for any $\psi_0$ that satisfies the  exponential bound eq.~ \eqref{eq:Amu} and any $t>0$,  
$$M_{\tau t }\left ( \frac{\vec{z}}{\sqrt{\tau}};\psi_0 \right ) \xrightarrow[]{\tau \rightarrow \infty} \norm{\psi_0}_{\ell^2}^2 \e^{-\frac{t}{2}  \vec{z} \cdot \vec{D} \vec{z}}$$
uniformly as $\vec{z}$ ranges over compact subsets of $\mathbb{C}^d$.
\end{thm}

From Theorem \ref{thm:CF}, the diffusive scaling for position moments, eq.\ \eqref{eq:diffusivescalingconjecture}, follows from the following ``central limit theorem.''
\begin{cor}[Central limit theorem for position averages]\label{cor:Diffusivescaling} 
Let $f:\R^d \rightarrow \R$ be a polynomially bounded continuous function.  If $\psi_0$ satisfies the exponential bound eq.~\eqref{eq:Amu}, then
$$\lim_{t\rightarrow \infty}  \sum_{x\in \Z^d} f\left (\frac{x}{\sqrt{t}} \right )
\Ev{\abs{\psi_t(x)}^2} \ = \ \norm{\psi_0}^2_{\ell^2} \int_{\R^d} f(\vec{x}) \frac{1}{\sqrt{\pi^d \det \vec{D}} } \e^{-\frac{1}{2}\vec{x} \cdot \vec{D}^{-1} \vec{x} } \di \vec{x},$$
with $\vec{D}$ as in Theorem \ref{thm:CF}. In particular, diffusive scaling, eq.\ \eqref{eq:diffusivescalingconjecture}, holds.
\end{cor}
\begin{rem*} If $\norm{\psi_0}^2=1$, then this really is a central limit theorem for the family of random variables $X_t$ where $\bb{P}(X_t = x) = \Ev{\abs{\psi_t(x)}^2}$.  The result is that $\nicefrac{X_t}{\sqrt{t}}$ converges in distribution to an $\R^d$-valued Gaussian variable with covariance $\vec{D}$.
\end{rem*}

\begin{proof}
To begin, note that the result holds for $f(x) = |x|^{2n}$ with $n=1,2,\ldots$.  Indeed, uniform convergence of analytic functions, as in Theorem \ref{thm:CF}, implies (uniform) convergence of derivatives to any order, so
\begin{multline*}\frac{1}{t^n} \sum_{x\in \Z^d} |x|^{2n} \Ev{\abs{\psi_t(x)}^2} \ = \ \left .\left ( - \Delta_{\vec{k}} \right )^n M_t\left ( \frac{\vec{k}}{\sqrt{t}};\psi_0 \right ) \right |_{\vec{k}=0} \\ \xrightarrow[]{t\rightarrow \infty} \
\norm{\psi_0}^2 \left . \left ( - \Delta_{\vec{k}} \right )^n \e^{-\frac{1}{2} \vec{k}\cdot \vec{D}\vec{k}} \right |_{\vec{k}=0} \ = \ \frac{\norm{\psi_0}^2}{\sqrt{\pi^d \det \vec{D}}}\int_{\R^d} \abs{\vec{x}}^{2n} \e^{-\frac{1}{2} \vec{x}\cdot \vec{D}^{-1} \vec{x}} \di \vec{x},
\end{multline*}
where $\Delta_{\vec{k}}$ denotes the Laplacian $\sum_j \nicefrac{\partial^2}{\partial k_j^2}$.   In particular, note that
\begin{equation}\label{eq:uniformlybounded}
\abs{ \left ( -\Delta_{\vec{k}} \right )^n M_t \left ( \frac{\vec{k}}{\sqrt{t}};\psi_0 \right ) } \ \le \  \sum_{x\in \Z^d} \frac{|x|^{2n}}{t^n} \Ev{|\psi_t(x)|^2} \ \le \ C_{n} \ < \ \infty
\end{equation}
for each $n$, uniformly in $t>0$ and $\vec{k}\in \R^d$.

To prove the general statement it clearly suffices to prove
$$\lim_{t\rightarrow \infty}  \sum_{x\in \Z^d} g\left (\frac{x}{\sqrt{t}} \right ) \left ( 1 + \frac{|x|^{2n}}{t^n} \right )
\Ev{\abs{\psi_t(x)}^2} \ = \ \norm{\psi_0}^2_{\ell^2} \int_{\R^d} g(\vec{x}) \frac{\left ( 1 + |\vec{x}|^{2n} \right )}{\sqrt{\pi^d \det \vec{D}} } \e^{-\frac{1}{2}\vec{x} \cdot \vec{D}^{-1} \vec{x} } \di \vec{x}$$
for arbitrary $g\in C_0(\R^d)$. By an approximation argument, using eq.\ \eqref{eq:uniformlybounded}, it suffices to prove this under the additional assumption that the Fourier transform $\widehat{g}\in L^1(\R^d)$ (such functions are uniformly dense in $C_0(\R^d)$). If $\widehat{g}\in L^1(\R^d)$ then
$$\sum_{x\in \Z^d} g\left (\frac{x}{\sqrt{t}} \right ) \left ( 1 + \frac{|x|^{2n}}{t^n} \right )\Ev{\abs{\psi_t(x)}^2} \ = \  \frac{\norm{\psi_0}_{\ell^2}^2 }{(2\pi)^d} \int_{\R^d} \widehat{g}(\vec{k}) \left ( 1 + (-\Delta_{\vec{k}})^n \right ) M_t\left ( \frac{\vec{k}}{\sqrt{t}} ;\psi_0 \right ) \di \vec{k}.$$
Let $\epsilon >0$. By eq.\ \eqref{eq:uniformlybounded} and since $\widehat{g}\in L^1$, we may find a compact set $F_{\epsilon}$ so that
\begin{multline*}\sum_{x\in \Z^d} g\left (\frac{x}{\sqrt{t}} \right ) \left ( 1 + \frac{|x|^{2n}}{t^n} \right )\Ev{\abs{\psi_t(x)}^2} \\ = \   \frac{\norm{\psi_0}_{\ell^2}^2 }{(2\pi)^d} \int_{F_{\epsilon}} \widehat{g}(\vec{k}) \left ( 1 + (-\Delta_{\vec{k}})^n \right ) M_t\left ( \frac{\vec{k}}{\sqrt{t}} ;\psi_0 \right ) \di \vec{k} \ + \ \mathcal{O}(\epsilon).
\end{multline*}
Since the derivatives of $M_t(\nicefrac{\vec{k}}{\sqrt{t}};\psi_0)$ converge uniformly on compact sets we conclude, by taking $t\rightarrow \infty$  and then $\epsilon \rightarrow 0$, that 
\begin{align*} \lim_{t\rightarrow \infty} \sum_{x\in \Z^d}  g\left (\frac{x}{\sqrt{t}} \right ) \left ( 1 + \frac{|x|^2n}{t^n} \right ) \Ev{\abs{\psi_t(x)}^2}
\ &= \ \frac{\norm{\psi_0}_{\ell^2}^2 }{(2\pi)^d} \int_{\R^d} \widehat{g}(\vec{k}) \left ( 1 + (-\Delta_{\vec{k}})^n \right ) \e^{-\frac{1}{2}  \vec{k} \cdot D \vec{k}} \di \vec{k}   \\
\ &= \ \norm{\psi_0}_{\ell^2}^2 \int_{\R^d} g(\vec{x})  \frac{\left ( 1 + \abs{\vec{x}}^{2n} \right ) }{\sqrt{(2\pi)^2 \det \vec{D}}} \e^{-\frac{1}{2}  \vec{x} \cdot D^{-1} \vec{x}} \di \vec{x}  \end{align*}
by Plancherel's formula. 
\end{proof}

\subsection{Organization}
The rest of the paper is organized in sections as follows:
\begin{itemize}
\item[\S\ref{sec:analytic}] \underline{\nameref{sec:analytic}}. Here we prove analyticity of the characteristic function using a general result on the finite group velocity for exponentially bounded hopping.
\item[\S\ref{sec:Feynman-Kac}]\underline{\nameref{sec:Feynman-Kac}}. In this section we review the formula from \cite{Kang2009b} that relates the characteristic function for real $\vec{z}$ to matrix elements of a semi-group and extend that formula to complex $\vec{z}$.
\item[\S\ref{sec:spectral}] \underline{\nameref{sec:spectral}}. The basic technical estimates bounding the location of the spectrum of the generators of the semi-groups of \S\ref{sec:Feynman-Kac} are stated here.
\item[\S\ref{sec:proof}]\underline{\nameref{sec:proof}}.
\item[] \hspace{-.5in} \underline{Appendix}
\item[\S\ref{sec:abstract}]\underline{\nameref{sec:abstract}}. The technical estimates of \S\ref{sec:spectral} are based on two abstract operator lemmas stated and proved here.
\item[\S\ref{sec:FSOPproof}]\underline{\nameref{sec:FSOPproof}}.
\end{itemize}

\subsection*{A word on notation} Given an operator $A$ on $\ell^2(\Z^d)$ we will denote the operator kernel of $A$ by $A(x,y)$.  That is
$$ A(x,y) = \ipc{\delta_x}{ A\delta_y}$$
so that
$$A\psi(x) \ = \ \sum_{y} A(x,y)\psi(y).$$
As above, we will use $I$ to denote the identity operator so that $I(x,y)$ is the Kronecker delta $\delta_{x,y}$.
If the operator $A$ depends on some other parameters, as $A(t_1,t_2,...)$ we will denote the kernel by $A(t_1,t_2,...;x,y)$.

\section{Finite group velocity and analyticity for the CF}\label{sec:analytic}
Analyticity of the CF (Lemma \ref{lem:CF}) follows from a \emph{finite group velocity} estimate for equations of the form eq.\ \eqref{eq:MSE} with exponentially bounded hopping terms. 
\begin{lem}[Finite Group Velocity]\label{lem:FSOP}
Let $U(t,s)$ be the unique solution to \begin{equation}\label{eq:propagator}i \partial_t U(t,s)= H(t) U(t,s) , \quad U(s,s)=I
\end{equation}
where the time dependent Hamiltonian $H(t)$ satisfies 
\begin{equation}\label{eq:expbound} v := \sup_{t}\sup_{y} \sum_{x\neq y}\e^{m|x-y|} \abs{H(t;x,y)} < \infty.
\end{equation}
Then
\begin{equation}\label{eq:fsop}\sum_{x} e^{m|x-y|} \abs{U(t,s;x,y)} \ \le \ e^{v|t-s|}.
\end{equation}
\end{lem}
\begin{rem*}The operator $U(t,s)$ is called the \emph{propagator} for the time dependent equation
$$i\partial_t \psi_t = H(t) \psi_t,$$
since it relates the solution at different times $\psi_t = U(t,s)\psi_s$.  If $H(t)$ is self-adjoint for every $t$, as in eq.~\eqref{eq:MSE}, then $U(t,s)$ is unitary.  However, the estimate eq.~\eqref{eq:fsop} does not use unitarity.  

Lem.\ \ref{lem:FSOP} is well known \textemdash \ essentially it is the simplification to this context of a more general Lieb-Robinson bound valid for lattice spin systems, see \cite{Nachtergaele2006} and references therein.  For completeness we include a short   proof below in Appendix \ref{sec:FSOPproof}.  The result may be interpreted as stating that (up to exponentially small tails) the propagator  $U(t,s;x,y)$ is negligible between sites $x,y$ with $|x-y| > v |t-s|$.  Thus the group velocity of the wave  is no larger than $v$.   Note that the upper bound $v$ on the group velocity is \emph{completely insensitive} to the diagonal terms $H(t;x,x)$. 
\end{rem*}

Lemma \ref{lem:CF} is a consequence of the following corollary of the Finite Group Velocity Lemma:
\begin{lem}\label{lem:expbound} If $\psi_0$ is exponentially bounded in the sense of eq.~\eqref{eq:Amu}, then for all $t>0$
$$\sup_{x} \e^{\mu|x|}\Ev{\abs{\psi_t(x)}} \le A_\mu \e^{vt} $$ 
where $A_\mu$ denotes the left hand side of eq.~\eqref{eq:Amu}. In particular, there is $C <\infty$ so that for $\lambda<\mu$ we have
\begin{equation}\label{eq:expsum} \sum_{x} \e^{\lambda |x|}\Ev{\abs{\psi_t(x)}^2} < \frac{Ce^{vt}}{(\mu-\lambda)^d}
\end{equation}
and the conclusions of Lemma \ref{lem:CF} hold.
\end{lem}
\begin{proof}
Note that 
$$\psi_t(x) = \sum_y U(t,0;x,y) \psi_0(y)$$
where $U(t,s)$ is the unitary propagator for eq.~ \eqref{eq:MSE}. Thus 
$$\e^{\mu|x|}\abs{\psi_t(x) } \ \le \ A_\mu \sum_{y} \e^{\mu |x-y|} \abs{U(t,0;x,y)} \ \le \ A_\mu \e^{v t}$$
since $\mu \le m$.  Since $|\psi_t(x)|^2 \le \|\psi_0\|_{\ell^2} |\psi_t(x)|$ (by unitarity), eq.~\eqref{eq:expsum} follows. 

Eq.~ \eqref{eq:expsum} shows that the series defining $M(\vec{z};\psi_0)$ converges uniformly and absolutely on compact subsets of $R_\mu$.  Since the terms of the series are analytic, the results claimed in Lemma \ref{lem:CF} follow from standard complex analysis.
\end{proof}

\section{Feynman-Kac-Pillet Formula}\label{sec:Feynman-Kac}
The key to the proof of diffusion for the CF (Theorem \ref{thm:CF}) lies in the following ``Feynman-Kac-Pillet formula'' 
\begin{equation}\label{eq:Kangetal}
M_t(\vec{z};\psi_0) = \ipc{\delta_0 \otimes 1}{\e^{-t \widehat{L}_{\vec{z}}} \widehat{\rho}_{\vec{z}}\otimes 1}_{\Hi}, \quad \vec{z}\in R_\mu,
\end{equation}
linking the characteristic function to a matrix element of a contraction semi-group on the ``augmented'' Hilbert space
\begin{equation}
\Hi \ := \ L^2(\Z^d\times \Omega) 
\end{equation}
with the product of counting measure on $\Z^d$ and $dP(\omega)$ on $\Omega$.  We use $\psi\otimes f$ to denote the product function
$$\psi\otimes f (x,\omega) \ = \ \psi(x) f(\omega).$$
We denote the constant function on $\Omega$ taking value $1$ by $1$, so that $\psi\otimes 1(x,\omega)=\psi(x)$. In eq.~\eqref{eq:Kangetal} the generator of the semigroup is
\begin{equation}\label{eq:L}
\widehat{L}_{\vec{z}}  \ = \ \im \widehat{K}_{\vec{z}} + \im \widehat{V} + B
\end{equation}
with ``hopping term''
$$ \widehat{K}_{\vec{z}}\phi(x,\omega)= \sum_{\zeta \in \Z^d} h(\zeta) \left [ \phi(x-\zeta,\omega) - \e^{-\im \vec{z}\cdot \zeta}\phi(x-\zeta,\sigma_\zeta \omega) \right ],$$
``potential term''
$$\widehat{V}\phi(x,\omega) = \left ( v_x(\omega) - v_0(\omega) \right ) \phi(x,\omega),$$ 
and ``dissipative term'' $B$ given by the Markov generator acting on product functions as $B(\psi\otimes f)=\psi\otimes (Bf)$.
The initial state is given by $\widehat{\rho}_{\vec{z}} \otimes 1$, with 
$$\widehat{\rho}_{\vec{z}}(x) = \sum_{y\in \Z^d} \e^{\im \vec{z}\cdot y} \psi_0(x+y) \overline{\psi_0(y)}.$$ 
If $\psi_0$ satisfies the exponential bound eq.\ \eqref{eq:Amu}, then $\vec{z} \mapsto \widehat{\rho}_{\vec{z}}$ is an analytic map from $R_\mu$ into $\ell^1(\Z^d)\subset \ell^2(\Z^d)$.  

For $\vec{z}\in \mathbb{R}^d$, eq.~\eqref{eq:Kangetal} is \cite[Eq. 3.18]{Kang2009b}, which was proved in \cite{Kang2009b} by making use of a Feynman-Kac type formula due to Pillet \cite{Pillet1985} for the mean density matrix $\Ev{\psi_t(x)\overline{\psi_t(y)}}$. For $\vec{z} \in R_\mu \setminus \R^d$,  eq.~\eqref{eq:Kangetal} essentially follows  by analytic continuation. 

To make the analytic continuation argument precise, it useful to recall that, in the terminology of Kato \cite{Kato1995}, an \emph{analytic family of type (B)}  is a family $T(z)$ of closed maximally sectorial operators depending on a complex parameter $z$ such that the corresponding family of closed sectorial forms,
$$\mathfrak{t}(z)[u] \ = \ \ \text{form closure of }\ipc{u}{T(z)u},$$
has a common domain $\mathcal{Q}$ and is analytic, in the sense that $\mathfrak{t}(z)[u]$ is an analytic function of $z$ for each $u\in \mathcal{Q}$. An analytic family of bounded operators is an analytic family of type (B) for which the operators $T(z)$ are all bounded  and thus the common form domain is the entire Hilbert space $\mathcal{H}$.  These notions extend directly to functions of several complex variables.
\begin{lem} As $\vec{z}$ ranges over $R_\mu$,
\begin{enumerate}
\item $\widehat{K}_{\vec{z}}$ is an analytic family of bounded operators,
\item $\widehat{L}_{\vec{z}}$ is an analytic family of type (B) with common form domain equal to the form domain $\mathcal{Q}(B)$ of $B$, and
\item the identity eq.~\eqref{eq:Kangetal} holds whenever $\phi_0$ satisfies eq.\ \eqref{eq:Amu}.
\end{enumerate}
\end{lem}

\begin{proof}
Boundedness of $\widehat{K}_\vec{z}$ for $\vec{z}\in R_\mu$ follows directly from the exponential bound on the hopping terms.  Explicitly,
\begin{equation}\label{eq:Kbound}\norm{\widehat{K}_{\vec{z}}} \ \le \ \sum_{\zeta} \e^{|\Im \vec{z}| \zeta } \abs{h(\zeta)}.
\end{equation}
The operator $\widehat{V}$ is also bounded, $\|\widehat{V}\| \le 2 \|v_0\|_{L^\infty}$. Thus the expression eq.~\eqref{eq:L} defines a maximally sectorial operator with form domain $\mathcal{Q}(B)$.   Analyticity of $\widehat{K}_{\vec{z}}$ and $\widehat{L}_{\vec{z}}$ are evident from the explicit dependence on $\vec{z}$.   Since, as remarked above, eq.~\eqref{eq:Kangetal} holds for $\vec{z}\in \R^d$ and the semigroup $\e^{-t \widehat{L}_{\vec{z}}}$ is holomorphic in $\vec{z}$ (see \cite[Theorem 2.6]{Kato1995}), eq.~\eqref{eq:Kangetal} holds for all $\vec{z}\in R_\mu$ by $d$-applications of the identity theorem, one for analytic continuation in each of the variables $z_1,\ldots,z_d$. \end{proof}

\section{Spectral theory of the generators}\label{sec:spectral}
Our proof of Theorem \ref{thm:CF} relies on a spectral analysis of $\widehat{L}_{\vec{z}}$ for $\vec{z}$ in a neighborhood of $\vec{0}$.  Before proceeding, it will be useful to make some general comments about these operators.  

First note that, as the potential term $\widehat{V}$ is self-adjoint, its contribution to the operator $\widehat{L}_{\vec{z}}$ is anti-hermitian.  When $\vec{z}\in \R^d$,  $\widehat{K}_{\vec{z}}$ is also self-adjoint. We conclude that $\widehat{L}_{\vec{z}}$ is maximally accretive and, by the Lumer-Phillips Theorem, the semi-group $\e^{-t\widehat{L}_{\vec{z}}}$ is contractive. However, for $\vec{z} \not \in \R^d$, the operator $\widehat{K}_{\vec{z}}$ is not self-adjoint, $\widehat{L}_{\vec{z}}$ is not accretive and the semi-group $\e^{-t\widehat{L}_{\vec{z}}}$ is not contractive. Instead $\widehat{K}_{\vec{z}}^\dagger = \widehat{K}_{\vec{z}^*},$
and 
$$\mathrm{Im} \widehat{K}_{\vec{z}}\phi(z,\omega) \ := \ \frac{1}{2\im} \left [ \widehat{K}_{\vec{z}}-\widehat{K}_{\vec{z}}^\dagger\right ] \phi(x,\omega) \  = \  \im  \sum_{\zeta \in \Z^d} h(\zeta) \e^{-\im \mathrm{Re}\vec{x} \cdot \zeta} \sinh( \vec{y} \cdot \zeta) \phi(x-\zeta,\sigma_\zeta\omega),$$
where $\vec{z}=\vec{x}+\im \vec{y}$, $\vec{x}, \vec{y} \in \R^d$.
In particular,
$$\norm{\Im \widehat{K}_{\vec{z}}} \ \le \ \sum_{\zeta \in \Z^d} \abs{h(\zeta)} \abs{\sinh{ \vec{y} \cdot \zeta}} \ =: \ \alpha(\vec{y}),$$
and 
\begin{equation}\label{eq:numericalrange}\mathrm{Re} \ipc{\phi} {\widehat{L}_{\vec{z}} \phi}_{\Hi} \ \ge \ -\alpha(\vec{y}), \qquad \phi \in \mathcal{Q}(B).
\end{equation}
Thus the semi-group $\e^{-t\widehat{L}_{\vec{z}}}$ is exponentially bounded by
\begin{equation}\label{eq:normbound}
\norm{\e^{-t\widehat{L}_{\vec{z}}}} \ \le \ \e^{t\alpha(\vec{y})}, \qquad \vec{z}=\vec{x}+\im \vec{y}.
\end{equation}
Since $\alpha(\vec{y}) = \mathcal{O}(\abs{\vec{y}})$ as $|\vec{y}|\rightarrow 0$,  the rate of growth is at least slow for $\vec{z}$ near $\R^d$.

In essence the proof of Theorem \ref{thm:CF} hinges on showing that, while eq.~\eqref{eq:numericalrange} is sharp, eq.~\eqref{eq:normbound} is not, and in fact
\begin{equation}\label{eq:asymptotic}\e^{-t\widehat{L}_{\vec{z}}} \ = \ \e^{t\lambda(\vec{y})} Q(\vec{y}) + \mathcal{O}(\e^{-\gamma t}), \qquad \vec{z}=\vec{x}+\im \vec{y}, 
\end{equation}
where $Q(\vec{y})$ is a rank one operator, $\lambda(\vec{y})= \mathcal{O}(|y|^2)$ and $\gamma >0$. 

To obtain  eq.\ \eqref{eq:asymptotic} we need to examine the operators $\widehat{L}_{\vec{z}}$, in particular $\widehat{L}_{\vec{0}}$, more closely.  For this purpose, we introduce a decomposition of the Hilbert space $\Hi$:
$$\Hi \ = \ \Hi_0 \oplus \Hi_1, \quad \Hi_0 = \ell^2(\Z^d) \otimes 1, \quad \Hi_1 = \ell^2(\Z^d) \otimes L^2_0(\Omega).$$ 
Correspondingly let $P_0$, $P_1$ denote orthogonal projection onto $\Hi_0$, $\Hi_1$ respectively, i.e.,
$$P_0f(x) \ = \ \int_\Omega f(x,\omega) d\mu(\omega), \quad P_1 f(x,\omega) \ = \ f(x,\omega) - P_0f(x).$$
The subspaces $\Hi_0$ and $\Hi_1$ are invariant under $B$ and $\widehat{K}_{\vec{z}}$.  Indeed,
$$P_0 B \ = \ B P_0 \ = \ 0 , \quad \text{ and } \quad B \ = \ P_1 B \ = \ B P_1 \ \ge \ \frac{1}{T} P_1 ,$$
while
\begin{equation}\label{eq:Kz0}
P_0 \widehat{K}_{\vec{z}}f(x) \ = \ \widehat{K}_{\vec{z}} P_0f(x) \ = \ \sum_{\zeta\in \Z^d}h(\zeta) \left [ 1-\e^{-\im \vec{z}\cdot \zeta} \right ] P_0f(x-\zeta).
\end{equation}
On the other hand, the subspaces $\Hi_0$ and $\Hi_1$ are not invariant under the potential $\widehat{V}$. In fact, the non-degeneracy condition eq.\ \eqref{eq:vnondeg} implies that,
\begin{equation}\label{eq:vnondeg2}
  \norm{B^{-1} \widehat{V}  P_0 f\ }^2 \ = \  \sum_{x\neq 0}  \norm{B^{-1}(v_x - v_0)}^2 \abs{P_0f(x)}^2 \ \ge\  \chi \sum_{x\neq 0} \abs{P_0f(x)}^2,
\end{equation}
so $P_1 \widehat{V} P_0 \neq 0$. However, because $\Ev{v_x-v_0}=0$, 
$$P_0 \widehat{V} P_0 \ = \ 0.$$

To summarize, with respect to the decomposition $\Hi = \Hi_1 \oplus \Hi_2$, the operator $\widehat{L}_{\vec{z}}$  has the following block matrix form
\begin{equation}\label{eq:blockmatrix}
\widehat{L}_{\vec{z}} \ = \ \begin{pmatrix} \im \widehat{K}_{\vec{z}}^{00} & \im \widehat{V}^{01} \\
\im \widehat{V}^{10} &  \im \left (\widehat{K}_{\vec{z}}^{11} + \widehat{V}^{11} \right ) + B
\end{pmatrix}
\end{equation}
where $A^{ij} = P_i A P_j$ for an operator $A$ and $i,j\in \{0,1\}$. This block matrix form allows us to prove the following Lemma about the spectrum of $\widehat{L}_{\vec{0}}$.
\begin{figure}																%%
\[																		%%
\begin{tikzpicture}[scale=1]									
%fill the spectrum in grey												%%
%\fill [gray!40] (8.5,3) --(3,1.5) -- (3,-1.5) -- (8.5,-3);
\fill [gray!40 ] (8.5,2.5) --(3,1.5) -- (3,-1.5) -- (8.5,-2.5);
\fill [pattern color = blue, semitransparent, pattern=horizontal lines] (8.5,2.75) --(2.75,1.75) -- (2.75,-1.75) -- (8.5,-2.75);
%outline spectrum														%%
\draw[thick] (8.5,2.5) --(3,1.5) -- (3,-1.5) -- (8.5,-2.5);	
\draw[blue,thick] (8.5,2.75) --(2.75,1.75) -- (2.75,-1.75) -- (8.5,-2.75);			
%makes and label vertical axis											
\draw[style=very thick, ->] (0,-3)--(0,3) node[above]{{$\Im w$}};  
%makes and label horizontal axis
\draw[style=very thick, ->] (-1.75,0)--(9,0) node[right]{$\Re w$};	
%scale axes		
\draw[style = dashed] (3,3) -- (3,-3) ;
\draw[<-](3.1,-2.8) -- (3.6, -2.8) node[right]{$\Re z = \frac{\gamma}{T}$};
\draw[style = dashed] (2.75,-3) -- (2.75,3);
\draw[<-](2.65,-2.8) -- (2.05,-2.8) node[left]{$\Re z = \frac{\gamma'}{T}$};																										%%
%label numerical range														%%
\draw[<-] (7.5,.5) -- (8.7,1)node[right, yshift=.125cm]{$\Sigma_+^{(\vec{0})}$};		
\draw[blue,<-] (7.5,2.45) -- (8.7,3)node[right, yshift=.125cm]{$\Sigma_+^{(\vec{z})}$};
					
\fill [color = black] (0, 0) circle (2pt);											%%
\draw [ <-]  (-.1, .05) --  (-1.1,.6) node[above, xshift=-.25cm] {$E\left(\vec{0}\right)=0$};	

\fill [color = blue] (0.5, -0.25) circle (2pt);
\draw [color = blue, <-]  (0.4, -0.3) --  (-1.1,-0.85) node[below, xshift=-.25cm] {$E\left(\vec{z}\right)$};												%%
\end{tikzpicture}											
\]	
\caption{\underline{Location of the spectrum of $\widehat{L}_{\vec{z}}$}. The operator has a non-degenerate eigenvalue  $E(\vec{z})$, while the rest of the spectrum is contained in the set $\Sigma_+^{(\vec{z})}$ shown in blue stripes. When $\vec{z}=0$ the eigenvalue $E(\vec{0})=0$.\label{fig:L0}}	
\end{figure}
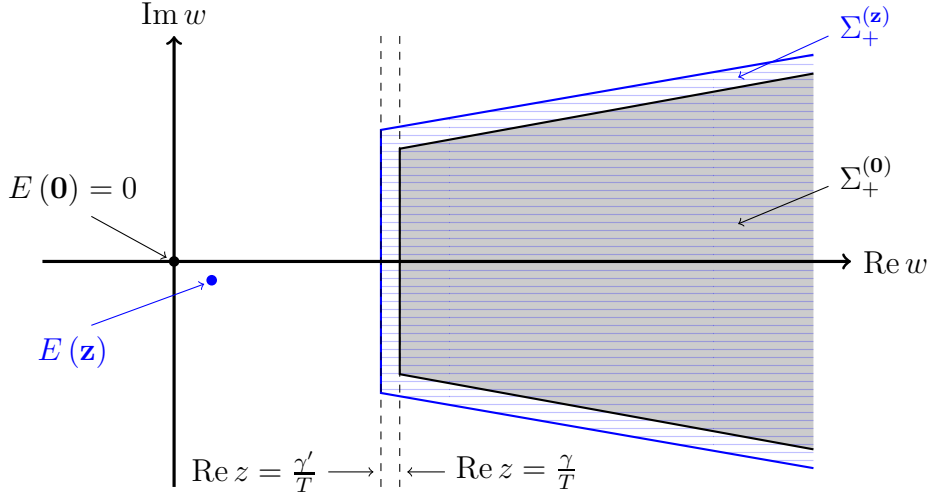

\begin{lem}\label{lem:L0} The operator $\widehat{L}_{\vec{0}}$ satisfies
\begin{enumerate}
\item $\widehat{L}_{\vec{0}} \delta_0 \otimes 1 \ = \ \widehat{L}_{\vec{0}}^\dagger \delta_0 \otimes 1 \ = \  0.$
\item $\widetilde{\Hi} = \{\delta_0\otimes 1\}^\perp$ is invariant under $\widehat{L}_{\vec{0}}$ and the spectrum of the operator restricted to this subspace is contained in a set $\Sigma_+^{(\vec{0})}$ of the form
\begin{equation} \setb{ w\in \C }{ \Re w \ge \frac{\gamma}{T}, \ \abs{\Im w} \le p \left [ \Re w + c \right ] }\label{eq:Sigma+}
\end{equation}
with $p,c >0$ and $0 <\gamma <1$.
\item There is $C <\infty$ so that if $w \not \in  \{0\} \cup\Sigma_+^{(\vec{0})}$, then
$$ \norm{\left (\widehat{L}_{\vec{0}} - w I\right )^{-1}} \ \le \ \frac{C}{\dist\left (w,\{0\}\cup \Sigma_+^{(\vec{0})} \right ) }.$$
\end{enumerate}
\end{lem}
\begin{rems*}1) See figure \ref{fig:L0}. 2)
The first identity is a closely related to the unitarity of the evolution, since
$$\norm{\psi_0}^2 \ = \ \Ev{ \norm{\psi_t}^2} \ = \ \ipc{\delta_0\otimes 1}{\e^{-t \widehat{L}_{\vec{0}}} \widehat{\rho}_{0}}_{\Hi}.$$
However, it may also be proved directly as below.
\end{rems*}

Lemma \ref{lem:L0} (without the estimate on the resolvent) is essentially \cite[Lemma 3]{Kang2009b}.  For completeness, we prove it below in \S\ref{sec:abstract}.  Simple perturbation theory now implies that the same qualitative picture holds for $\widehat{L}_{\vec{z}}$  once $\vec{z}$ is sufficiently small:
\begin{lem}\label{lem:Lz} There is $\epsilon >0$ so that if $|\vec{z}| < \epsilon$, then the spectrum of $\widehat{L}_{\vec{z}}$ may be decomposed as
\begin{equation}\label{eq:specdecomp}
\sigma(\widehat{L}_{\vec{z}}) \ = \ \{E(\vec{z})\} \cup \sigma_+^{(\vec{z})}
\end{equation}
where
\begin{enumerate}
\item $E(\vec{z})$ is a non-degenerate eigenvalue that depends analytically on $\vec{z}$ and is located in a disk of radius $\nicefrac{\gamma}{4T}$ around $0$.
\item $\sigma_+^{\vec{z}}$  is a subset  of a set $\Sigma_+^{(\vec{z})}$ of the form eq.\ \eqref{eq:Sigma+} with modified values of the constants,
\begin{equation}\label{eq:Sigma+2}\sigma_+^{\vec{z}} \ \subset \ \Sigma_+^{(\vec{z})} \ := \ \setb{ w\in \C }{ \Re w \ge \frac{\gamma'}{T}, \ \abs{\Im z} \le p' \left [ \Re w + c' \right ] }
\end{equation}
where $$\gamma' \ = \ \gamma - \mathcal{O}(|\vec{z}|), \quad p'=p, \quad \text{ and }c'= c + \mathcal{O}(|\vec{z}|),$$
with $\gamma,p,c$ as in eq.\ \eqref{eq:Sigma+}.
\end{enumerate}
Furthermore, there is $C <\infty$ such that for $|\vec{z}|<\epsilon$ and $w\not \in \{E(\vec{z}) \} \cup \Sigma_+^{(\vec{z})}$,
\begin{equation}\label{eq:resolventbound2} \norm{(\widehat{L}_{\vec{z}} - w I )^{-1}} \ \le \ \frac{C}{\dist\left (w,\{E(\vec{z})\}\cup \Sigma_+^{(\vec{z})} \right ) }.
\end{equation}
\end{lem}

\begin{proof}
Note that $\norm{\widehat{L}_{\vec{z}} - \widehat{L}_{\vec{0}}}\le \beta |\vec{z}|$ for some $\beta >0$. It follows from part 3 of Lemma \ref{lem:L0} that if $\dist\left (w, \{0\} \cup \Sigma_+^{(\vec{0})} \right )>  C \beta |\vec{z}|$, then $w$ is in the resolvent set of $\widehat{L}_{\vec{z}}$ and  
\begin{equation} \norm{\left ( \widehat{L}_{\vec{z}} - w I \right )^{-1}} \ \le \ \frac{C}{\dist\left (w, \{0\} \cup \Sigma_+^{(\vec{0})} \right ) - C \beta |\vec{z}| }.\label{eq:resolventbound3}
\end{equation}

Take $\epsilon = \frac{\gamma}{4C\beta T}$. Then for $|\vec{z}| <\epsilon$, we have
$$\sigma(\widehat{L}_{\vec{z}}) \subset \setb{w}{|w|<\frac{\gamma}{4T}} \cup \setb{w}{\dist(w,\Sigma_{+}^{(\vec{0})}) < \frac{\gamma}{4T}}$$
and 
$$\norm{\left ( \widehat{L}_{\vec{z}} - w I \right )^{-1}} \ \le \ \frac{4C}{\gamma T }.$$
on the circle $|w|=\nicefrac{\gamma}{2T}$. Thus the integral
$$Q_{\vec{z}} \ = \ \frac{1}{2\pi \im} \oint_{\abs{w}= \frac{\gamma}{2T}} \left ( w I - \widehat{L}_{\vec{z}} \right )^{-1} \di w$$
defines a Riesz spectral projection for $\widehat{L}_{\vec{z}}$.  In particular, $Q_{\vec{0}}$ is the  orthogonal projection on the span of $\delta_0 \otimes 1$. By analytic continuation, we see that $Q_{\vec{z}}$ is rank one throughout $\{|\vec{z}| <\epsilon\}$. Thus $\widehat{L}_{\vec{z}}$ has a single non-degenerate eigenvalue inside $\{|w|<\nicefrac{\gamma}{4T}\}$. By standard results of pertrubation theory (see \cite[Chapter 7]{Kato1995}) the corresponding eigenvalue $E(\vec{z})$ is analytic in $\vec{z}$.

Now let  $\Sigma_+^{(\vec{z})}$ be of the form eq.\ \eqref{eq:Sigma+2} with
$$\frac{\gamma'}{T} = \frac{\gamma}{T} - C \beta  |\vec{z}|,\quad c' =  c+ \frac{C\beta|\vec{z}|}{p},\quad p'=p. $$
Then for $w\not \in \Sigma_+^{(\vec{z})}$, we have $\dist(w,\Sigma_+^{(\vec{0})})> C\beta |\vec{z}|$ and so $w$ is in the resolvent set unless $w=E(\vec{z})$.  Thus we have proved eq.\ \eqref{eq:specdecomp}.

Finally, we must estimate the norm of the resolvent.  Suppose $|\vec{z}|<\epsilon$, $w\neq E(\vec{z})$ and $w \not \in \Sigma_+^{(\vec{z})}$. We consider two cases
\begin{enumerate}[leftmargin=.275in]
\item   If $|w| \ge \nicefrac{3\gamma}{8T}$ then  
$$\dist(w,E(\vec{z})) \ \le \ |w| + \frac{\gamma}{4T} \  \le \ \frac{5}{3} |w| \ \le \ 5 \left ( |w| - \frac{\gamma}{4 T} \right ) \ \le \  5 \left  (|w| - C \beta |\vec{z}|\right ).$$
Thus
$$ \dist(w, \{E(\vec{z})\}\cup \Sigma_+^{(\vec{z})}) \ \le \  5 \left (\dist\left (w, \{0\} \cup \Sigma_+^{(\vec{0})} \right ) - C \beta |\vec{z}| \right ),$$
and the desired estimate, eq.\ \eqref{eq:resolventbound2}, follows from eq.\ \eqref{eq:resolventbound3}. 
\item If $|w| < \nicefrac{3\gamma}{8T}$, we use the identity
$$(\widehat{L}_{\vec{z}}-w I)^{-1} \ = \ \frac{1}{E(\vec{z})-w} Q_{\vec{z}} + \left ( \widehat{L}_{\vec{z}}-wI \right )^{-1} (I-Q_{\vec{z}}).$$
As the second term is analytic throughout $|w| \le \nicefrac{\gamma}{2T}$, we have
$$  (\widehat{L}_{\vec{z}} - w I)^{-1} (I- Q_{\vec{z}})
\ = \ \frac{1}{2\pi \im} \oint_{|w'|=\frac{\gamma}{2T}} \frac{1}{w'-w} \left ( \widehat{L}_{\vec{z}} - w' I\right )^{-1} (I- Q_{\vec{z}} )  \di w',$$
and thus
$$\norm{(\widehat{L}_{\vec{z}}- wI)^{-1}} \ \le \   \frac{1}{\abs{E(\vec{z}) - w}}\norm{Q_{\vec{z}}} + \frac{8T}{3\gamma}\frac{4C}{\gamma T} \left [ 1 + \norm{Q_{\vec{z}}}\right ] .$$
Since $\norm{Q_{\vec{z}}} \ \le \ \nicefrac{4 C}{\gamma T}$
for $|\vec{z}| < \epsilon$ and 
$$\dist(w, \{E(\vec{z})\}\cup \Sigma_+^{(\vec{z})}) \ = \ \dist(w,E(\vec{z})) \ \le \ \frac{3\gamma}{4T}$$
the desired estimate follows. \qedhere
\end{enumerate}
\end{proof}

\section{Convergence of the characteristic function \textemdash \ Proof of Theorem \ref{thm:CF}}\label{sec:proof}
We now turn to the proof of Theorem \ref{thm:CF}.  Let $t>0$ and $K\subset \C^d$ be compact. Let 
$$ F_\tau(\vec{z})  \ = \ M_{\tau t} \left (\frac{\vec{z}}{\sqrt{\tau}};\psi_0 \right ).$$
 We wish to show
$$\lim_{\tau\rightarrow \infty} \sup_{\vec{z}\in K} \abs{F_{\tau}(\vec{z}) - \norm{\psi_0}_{\ell^2}^2 \e^{-t \sum_{i,j} D_{i,j} z_i z_j}} \ = \ 0.$$
Since $F_\tau(\cdot)$ is defined on $\sqrt{\tau} R_\mu$ we eventually have $F_{\tau}(\cdot)$ defined on $K$ and it makes sense to talk about the convergence. 

By the Feynman-Kac-Pillet formula,
$$F_\tau(\vec{z}) \ = \ \ipc{\delta_0\otimes 1}{\e^{-\tau t \widehat{L}_{\nicefrac{\vec{z}}{\sqrt{\tau}}}} \widehat{\rho}_{\nicefrac{\vec{z}}{\tau}} }.$$
For large enough $\tau$ we have $|\vec{z}|<\epsilon \sqrt{\tau}$ with $\epsilon$ as in Lemma \ref{lem:Lz}.  Thus 
\begin{equation} F_\tau(\vec{z}) \ = \ \e^{-\tau t E(\nicefrac{\vec{z}}{\sqrt{\tau}})} \ipc{\delta_0\otimes 1}{Q_{\nicefrac{\vec{z}}{\sqrt{\tau}}}  \widehat{\rho}_{\nicefrac{\vec{z}}{\tau}} } + \norm{\e^{-\tau t \widehat{L}_{\nicefrac{\vec{z}}{\tau}}} (I-Q_{\nicefrac{\vec{z}}{\sqrt{\tau}}})} \norm{\psi_0}^2.\label{eq:almostthere}
\end{equation}
Because $\widehat{L}_{\nicefrac{\vec{z}}{\sqrt{\tau}}}$ is sectorial, it is holomorphic and the spectral bound eq.\ \eqref{eq:Sigma+2} implies that 
$$\norm{\e^{-\tau t \widehat{L}_{\nicefrac{\vec{z}}{\tau}}}} \ \le \ C \e^{-\tau t (\gamma - \mathcal{O}(\nicefrac{|\vec{z}|}{\sqrt{\tau}} )}.$$
(See \cite[\S IX.1.6]{Kato1995}.) Thus the second term in eq.\ \eqref{eq:almostthere} is negligible in the large $\tau$ limit.  To compute the limit of the first term, we note that
\begin{enumerate}
\item $\widehat{\rho}_{\vec{z}}$ and $Q_{\vec{z}}$ are continuous (in norm) at $\vec{z} =0$.
\item The derivatives of $E(\vec{z})$ vanish at $\vec{z}=0$ since
$$\left . \frac{\partial E(\vec{z})}{\partial z_i} \right |_{\vec{z}=0} \ = \ \im \ipc{\delta_0\otimes 1}{\left . \frac{\partial \widehat{K}_{\vec{z}}}{\partial z_i} \right |_{\vec{z}=0}\delta_0\otimes 1} 
\ = \ \sum_{\zeta\in \Z^d} h(\zeta) \zeta_i \ipc{\delta_0\otimes 1}{\delta_\zeta \otimes 1} \ = \ 0.$$
\item For the second derivatives we have
$$ \left . \frac{\partial^2 E(\vec{z})}{\partial z_i\partial z_j} \right |_{\vec{z}=0} \ = \ \left .\ipc{\delta_0\otimes 1}{  \frac{\partial^2 \widehat{K}_{\vec{z}}}{\partial z_i\partial z_j} \delta_0 \otimes 1} \right |_{\vec{z}=0} \ + \ 2 \Re  \left .\ipc{\frac{\partial \widehat{K}_{\vec{z}}}{\partial z_i} \delta_0\otimes 1}{\left [ \widehat{L}_{\vec{0}} \right ]^{-1} \frac{\partial \widehat{K}_{\vec{z}}}{\partial z_i} \delta_0\otimes 1} \right |_{\vec{z}=0},$$
where $[\widehat{L}_{\vec{0}}]^{-1}$ denotes the inverse of $\widehat{L}_{\vec{0}}$ restricted to $\{\delta_0\otimes 1\}^\perp$ (which is bounded by Lemma \ref{lem:L0}). The first term on the right hand side is zero by a similar calculuation as for the first derivatives, however the second term is non-zero and gives
$$ \left . \frac{\partial^2 E(\vec{z})}{\partial z_i\partial z_j} \right |_{\vec{z}=0} \ = \ 2 \Re\sum_{\zeta,\zeta'\in \Z^d}  \zeta_i \zeta_j' \overline{h(\zeta)}h(\zeta') \ipc{ \delta_\zeta \otimes 1}{\left [ \widehat{L}_{\vec{0}} \right ]^{-1}  \delta_{\zeta'} \otimes 1} \ =: \ \vec{D}_{i,j}.$$
\end{enumerate}
We conclude that
$$ \tau t E(\nicefrac{\vec{z}}{\sqrt{\tau}}) \ = \ t \vec{z}\cdot \vec{D} \vec{z}  \ + \ \mathcal{O}(\nicefrac{1}{\sqrt{\tau}}).$$
Theorem \ref{thm:CF} now follows from eq.\ \eqref{eq:almostthere} since $ Q_{\vec{z}}\widehat{\rho}_{\vec{z}} \ \xrightarrow[]{|\vec{z}| \rightarrow 0} \ \norm{\psi_0}^2 \delta_0 \otimes 1.$

\appendix

\section{Two abstract operator lemmas and the proof of lemma \ref{lem:L0}.}
\label{sec:abstract}
\numberwithin{lem}{section}
\setcounter{lem}{0}
To prove Lemma \ref{lem:L0} we will make use of two abstract operator lemmas:
\begin{lem}\label{lem:abstract1} Let $A$ be a boundedly invertible operator on a Hilbert space $\Hi$ and suppose that $z\in R(A)$, the resolvent set of $A$.  Then
$$ \norm{(A-zI)^{-1} f} \ \ge \ \frac{1}{1 + |z| \norm{A^{-1}}} \norm{A^{-1} f}$$
for all $f\in \Hi$.
\end{lem}

\begin{proof}
Let $f \in \Hi$.  Since $z \in R(A)$, we have $(A-zI)^{-1}f\in \mathcal{D}(A)$ and 
$$ (A -z I) (A-zI)^{-1} f \ = \ f.$$
Operating on both sides with $A^{-1}$, we find that
$$ (I - z A^{-1} ) (A-zI)^{-1} f \  = \ A^{-1} f.$$
Thus
$$ \norm{A^{-1} f} \ \le \ \left (1 + |z| \|A^{-1}\| \right ) \norm{(A-zI)^{-1} f}$$
as claimed.
\end{proof}
\begin{lem}\label{lem:abstract2} Let $\Hi_0$ and $\Hi_1$ be Hilbert spaces, $W:\Hi_0 \rightarrow \Hi_1$ a non-zero bounded linear map, and $A$ a maximally accretive operator on $\Hi_1$ with 
$$\Re \ipc{f}{Af} \ \ge \ \Delta \norm{f}_1^2, \quad \text{ for all }f\in \Hi_1,$$
where $\Delta$ is a positive real number. Consider the maximally accretive operator $L$ on $\Hi_0 \oplus \Hi_1$ given in block matrix form by
$$ L = \begin{pmatrix} 0  & - W^\dagger \\  W & A \end{pmatrix}$$
with domain $\Hi_0 \oplus \mathcal{D}(A)$. If for some $\alpha  >0$ we have \begin{equation}\label{eq:lowerbound}\norm{A^{-1} W f} \ge \alpha \frac{\norm{W}}{\Delta} \norm{f} \quad \text{ for all }\quad  f \in \Hi_0,
 \end{equation}
then the spectrum of $L$ is contained in the half plane
\begin{equation}\label{eq:halfplane}
\sigma(L) \subset \setb{z}{ \Re z \ge \gamma \Delta } ,
\end{equation}
where 
$$\gamma \ > \  \frac{1}{1 + \phi^2}, \qquad \phi \ = \ \frac{1 }{\alpha } \left ( 2 \frac{\Delta}{\norm{W}}  +  \frac{9}{8} \frac{\norm{W}}{\Delta}  \right ) .$$
Furthermore, there is a finite constant $C$ such that
\begin{equation}\label{eq:normboundresolvent}
\norm{(L-zI)^{-1}} \ \le \ \frac{C}{\gamma-\lambda}
\end{equation}
whenever $\Re z = \lambda <\gamma \Delta$.
\end{lem}
\begin{rems*}
1) Note that $\gamma >0$.  In particular, $\inf \setb{\Re z}{z\in \sigma(L)} >0$, even though $L$ is not \emph{strictly} accretive (since $0$ is in its  numerical range). 2) The theorem is non-perturbative: it holds without a smallness condition on $W$ \textemdash \ although $W$ must satisfy a bound of the form eq.~\eqref{eq:lowerbound}.  The conclusion eq.~\eqref{eq:halfplane} fails if $0$ is an eigenvalue of $W$, since then $0$ is an eigenvalue of $L$.
\end{rems*}

\begin{proof} First note that, since $\norm{A^{-1} W} \le \nicefrac{\norm{W}}{\Delta}$, we must have $\alpha \le 1$. Second note that it suffices to prove the lemma with $\Delta =1$. Indeed we may reduce to this case by multiplying $L$ by $\nicefrac{1}{\Delta}$.  Under this transformation $\alpha$ is unchanged but $\norm{W} \mapsto \nicefrac{\norm{W}}{\Delta},$ which is why the expression for $\gamma$ only depends on $W$ via the ratio $\nicefrac{\norm{W}}{\Delta}$.  

Henceforth we will assume that $\Delta =1$. In particular for $\Re z  < 1$ we have
\begin{equation}\label{eq:dissbound}
\norm{(A-zI)^{-1}} \ \le \ \frac{1}{1-\Re z}.
\end{equation}
Now consider the resolvent equation for $L$,
\begin{equation}\label{eq:resolvent}  \begin{pmatrix} g_0 \\ g_1 \end{pmatrix} \ = \ (L -zI) \begin{pmatrix} f_0 \\ f_1 \end{pmatrix} \ = \ \begin{pmatrix} -W^\dagger f_1 -z f_0 \\ W f_0 + (A-zI) f_1 \end{pmatrix}
\end{equation}
given $g_i\in \Hi_i$, $i=1,2$.  For $\Re z   < 1$,  we may solve for $f_1$ using eq.~\eqref{eq:dissbound}, 
\begin{equation}\label{eq:complementary} f_1 \ = \ (A-zI)^{-1} g_1 - (A-zI)^{-1} W f_0,
\end{equation}
so that eq.\ \eqref{eq:resolvent} reduces to the following equation for $f_0$,
\begin{equation}\label{eq:Schur} \left [ W^\dagger (A-zI)^{-1} W  - z I \right ] f_0 \ = \ g_0 + W^\dagger (A-zI)^{-1} g_1.
\end{equation}
In particular, for $z$ with $\Re z < 1$,  $z\in R(L)$ if and only if the bounded operator $S(z)$ in brackets on the left hand side of eq.~\eqref{eq:Schur} is invertible. 

To find a condition for invertibility of $S(z)$ we will investigate its numerical range. Fix $z=\lambda+\im \eta$ with $\lambda,\eta\in \R$, $\lambda <1$.  Given $f_0\in \Hi_0$, we have
\begin{equation}\label{eq:imaginary} \abs{\Im \ipc{f_0}{S(z) f_0}} \ \ge \ \abs{\eta} \norm{f_0}^2 - \norm{(A-zI)^{-1} } \norm{Wf_0}^2 \ \ge \ \left [ |\eta| - \frac{w^2}{1 -\lambda} \right ] \norm{f_0}^2,
\end{equation}
where $w=\norm{W}$.  Similarly,
\begin{multline}\label{eq:real} \Re \ipc{f_0}{S(z) f_0} \ \ge \  \Re \ipc{(A-zI)^{-1} W f_0}{ (A-z I) (A-zI)^{-1} W f_0} - \lambda \norm{f_0}^2 \\ \ge\  \left ( 1-\lambda \right ) \norm{(A-zI)^{-1} W f_0}^2 - \lambda \norm{f_0}^2
\ \ge \ \left [ \frac{\alpha^2 w^2 (1-\lambda)}{\left ( 1 + |z| \right )^2} - \lambda \right ] \norm{f_0}^2
\\ \ge \ \left [ \frac{\alpha^2 w^2}{\left ( 2 + |\eta| \right )^2} (1-\lambda) - \lambda  \right ] \norm{f_0}^2
\end{multline}
by Lem.\ \ref{lem:abstract1} and the assumed lower bound on $A^{-1} W$. 

We see from eqs.~(\ref{eq:imaginary}, \ref{eq:real}), that $S(z)$ is invertible provided either
\begin{equation}\label{eq:duality}\lambda < 1 - \frac{w^2}{|\eta|}  \quad \text{ or } \quad  \lambda < \frac{\alpha^2 w^2 }{\left ( 2 + |\eta| \right )^2 +\alpha^2 w^2}.
\end{equation}
The function on the right hand side of the left hand inequality is strictly increasing in $|\eta|$, while that on the right hand side of the right hand inequality is strictly decreasing.  Thus for any $z$ on the line $\Re z = \lambda$ at least one of the inequalities holds (and hence $S(z)$ is invertible) provided
\begin{equation}\label{eq:gamma}\lambda \ < \ \gamma \ := \  \frac{\alpha^2 w^2 }{ \left ( 2 + \eta_+ \right )^2     +\alpha^2 w^2} 
\end{equation}
where $\eta_+$ the unique positive number such that
\begin{equation}\label{eq:eta+}
 1 -\frac{w^2}{\eta_+} \ = \ \frac{\alpha^2 w^2 }{\left ( 2 + \eta_+ \right )^2 +\alpha^2 w^2} .
\end{equation}

We could find $\eta_+$ explicitly by solving this cubic equation, although the expression would be complicated.  To obtain a simple bound,  note that
$$
\frac{\alpha^2 w^2 }{\left ( 2 + \frac{9}{8} w^2  \right )^2 +\alpha^2 w^2} \ = \ \frac{1}{\frac{1}{\alpha^2} \left ( \frac{2}{w} + \frac{9}{8} w \right )^2 + 1} \ \le \ \frac{1}{ \left ( \frac{2}{w} + \frac{9}{8} w \right )^2} \ \le \  \frac{1}{9} \ =  \ 1-\frac{w^2}{\frac{9}{8} w^2},
$$
where we have used the fact that $\alpha \le 1$ and the inequality $a w + bw^{-1} \ge 2\sqrt{ab}$ for $w,a,b>0$. 
We conclude that $\eta_+ \le \nicefrac{9w^2}{8}$ and thus by eq.\ \eqref{eq:eta+} that 
$$\gamma  \ > \ \frac{1}{\frac{1}{\alpha^2} \left ( \frac{2}{w} + \frac{9}{8} w \right )^2 + 1}$$
as claimed.

It remains to estimate the norm of the resolvent $(L-zI)^{-1}$ for $\Re z < \gamma$.  To begin we note that there is a finite constant $C$ such that throughout this half plane
$$\norm{(L-zI)^{-1}} \ \le \ \frac{1}{1-\lambda} + C \norm{S(z)^{-1}}.$$
Indeed looking back at eqs.~(\ref{eq:complementary}, \ref{eq:Schur}) we see that 
$$\norm{f_0} \ \le \ \norm{S(z)^{-1}} \left ( \norm{g_0} + \frac{w}{1-\lambda} \norm{g_1} \right ) \ \le \left ( 1 + \frac{w}{1-\gamma} \right ) \norm{S(z)^{-1}} \norm{\begin{pmatrix} g_0 \\ g_1 \end{pmatrix}}$$
and  similarly
$$\norm{f_1} \ \le \ \frac{1}{1-\lambda} \left ( \norm{g_1} + w \norm{f_0} \right ) \ \le \ \left [ \frac{1}{1-\lambda} + \left ( \frac{w}{1-\gamma} + \frac{w^2}{(1-\gamma)^2} \right ) \norm{S(z)^{-1}} \right ] \norm{\begin{pmatrix} g_0 \\ g_1 \end{pmatrix}} .$$

Hence, to obtain eq.\ \eqref{eq:normboundresolvent}, it suffices to show that
\begin{equation}\label{eq:Schurbound}
\norm{S(z)}^{-1} \ \le \ \frac{C}{\gamma - \lambda}
\end{equation}
for $\Re z = \lambda <\gamma$. For this purpose, it is useful to observe that for $\lambda =\gamma$ and $\eta=\eta_+$ both inequalities in eq.~\eqref{eq:duality} saturate and the right hand sides of eqs.~(\ref{eq:imaginary}, \ref{eq:real}) are both zero,
$$ \eta_+ - \frac{w^2}{1- \gamma} \ = \ 0 \ = \ \frac{\alpha^2 w^2}{(2 + \eta_+)^2} (1-\gamma) - \gamma.$$
Now suppose $\lambda =\Re z <\gamma$.  If $|\eta| \geq \eta_+$ then 
$$ \frac{ \abs{\Im \ipc{f_0}{S(z)f_0}}}{ \norm{f_0}^2} \ \ge \ \left [ \eta_+ - \frac{w^2}{1-\lambda} \right ] \ \ge \ w^2 \left [ \frac{1}{1-\gamma} - \frac{1}{1-\lambda} \right ] \ \ge \ \frac{w^2}{(1-\gamma)^2} (\gamma-\lambda).$$
On the other hand if $|\eta| \leq \eta_+$ then
$$\frac{\Re \ipc{f_0}{S(z) f_0}}{ \norm{f_0}^2}  \ \ge \ \left [ \frac{\alpha^2 w^2}{(2 + \eta_+)^2 } (1-\lambda) -\lambda \right ] \ \ge \ \left [ \frac{\alpha^2 w^2}{(2 + \eta_+)^2 } (1-\gamma) -\lambda \right ]\ = \ \gamma -\lambda . $$
In either case we see that eq.\ \eqref{eq:Schurbound}, and thus eq.\ \eqref{eq:normboundresolvent}, holds.\end{proof}

We now turn to the
\begin{proof}[Proof of Lemma \ref{lem:L0}]
Since  $\widehat{K}_{\vec{0}}^{00}=0$ by eq.\ \eqref{eq:Kz0}, the block decomposition eq.\ \eqref{eq:blockmatrix} reduces in this case to 
$$\widehat{L}_{\vec{0}} = \begin{pmatrix} 0  & \im \widehat{V}^{01} \\
\im \widehat{V}^{10} &  \im \left (\widehat{K}_{\vec{0}}^{11} + \widehat{V}^{11} \right ) + B
\end{pmatrix}.
$$
We cannot apply Lemma \ref{lem:abstract2} directly, since 
$$\widehat{V}^{10} \delta_0 \otimes 1 \ = \ \widehat{V} \delta_0 \otimes 1 = \delta_0 \otimes (v_0 -v_0) \ = \ 0.$$
In particular, we see that  $\delta_0 \otimes 1$ is in the kernel of $\widehat{L}_{\vec{0}}$ and $\widehat{L}_{\vec{0}}^\dagger$  and thus that the subspace $\widetilde{\Hi} = \{\delta_0\otimes 1\}^\perp$ is invariant under $\widehat{L}_{\vec{0}}$ as claimed.  

Let $L=\left . \widehat{L}_{\vec{0}} \right |_{\widetilde{\Hi}}$ denote the  restriction of $\widehat{L}_{\vec{0}}$ to $\widetilde{\Hi}$.  We will apply Lemma \ref{lem:abstract2} to $L$ using the  decomposition $\widetilde{\Hi}= \widetilde{\Hi}_0\oplus \Hi_1$ where $$\widetilde{\Hi}_0  \ = \ P_0 \widetilde{\Hi} \ = \ (P_0 - Q_0) \Hi$$ 
with $Q_0$ the orthogonal projection onto $\delta_0\otimes 1$.
The operator $L$ has the block matrix form
$$ L \ = \ \begin{pmatrix}  0 & -W^\dagger  \\
 W & \widehat{L}_{\vec{0}}^{11}
 \end{pmatrix}, \quad W =  \left . \im \widehat{V} \right |_{\widetilde{\Hi}_0} , $$
 with a strictly accretive lower right-hand block
\begin{equation}
\Re \ipc{f}{\widehat{L}_{\vec{0}}^{11}f}\ = \  \Re \ipc{f}{Bf} \ \ge \frac{1}{T}\norm{f}^2, \quad f\in \Hi_1
\end{equation}
by the gap condition eq.\ \eqref{eq:gap}. Thus to apply Lemma \ref{lem:abstract2} we need only verify eq.\ \eqref{eq:lowerbound}:
\begin{align*}\norm{ \left (\widehat{L}_{\vec{0}}^{11} \right )^{-1} W f}
\ &= \  \norm{ (\widehat{K}_0^{11} + \widehat{V}^{11} + B)^{-1} Wf} \\ &= \ \norm{ \left (B^{-1} \widehat{K}_0^{11} + B^{-1} \widehat{V}^{11} + I \right )^{-1} B^{-1} Wf}  \\ &\ge \ \frac{1}{\norm{B^{-1}  \widehat{K}_0^{11}} + \norm{B^{-1} \widehat{V}^{11}} + 1} \norm{B^{-1} W f} \\ &\ge \ \frac{\chi}{\norm{B^{-1}  \widehat{K}_0^{11}} + \norm{B^{-1} \widehat{V}^{11}} + 1} \norm{f},
\end{align*}
where in the second to last inequality we used Lemma \ref{lem:abstract1} and in the last eq.~\eqref{eq:vnondeg2}.

Hence, by Lemma \ref{lem:abstract2}, we find that $\sigma(L)$ is contained in a half plane $\{\Re z \ge \nicefrac{\gamma}{T}\}$ with $\gamma>0$. On the other hand we have
$$\abs{\Im \ipc{f}{Lf}} \ \le \ \left [ \norm{\widehat{K}_{\vec{0}}} + \norm{\widehat{V}} \right ] \norm{f}^2 + \abs{\Im \ipc{f}{Bf}} .$$
By the sectoriality of $B$ eq. \eqref{eq:sector}, we find that the numerical range of $L$ and  spectrum of $L$ are contained in a set of the form
$$ \setb{z}{ \abs{\Im z } \le p \left [ \Re z + c \right ]}$$
for suitable $p,c >0$.  The intersection of this set and the half plane $\{\Re z \ge \nicefrac{\gamma}{T}\}$ gives the set $\Sigma_+^{(\vec{0})}$ of the lemma. Finally, the norm bound on the resolvent follows from eq.\ \eqref{eq:normboundresolvent} of Lemma \ref{lem:abstract2} together with the standard estimate
$$ \norm{\left (A-z \right)^{-1}} \ \le \ \frac{1}{\dist(z,\text{numerical range of $A$})}$$
valid for maximally sectorial operators.
\end{proof}

\section{Proof of the finite group velocity lemma}\label{sec:FSOPproof}
The first step in the proof of Lem.\ \ref{lem:FSOP} is to remove the diagonal terms, by changing to the ``interaction picture'' for the potential $V(t,x)=H(t;x,x)$.  That is we replace $U(t,s;x,y)$ by the kernel
$$W(t,s;x,y) = \e^{i\int_{t_0}^t V(r,x) dr} U(t,s;x,y)\e^{-i\int_{t_0}^s V(r,y) dr}$$
for some fixed, but arbitrary, $t_0$. Then $W(t,s;x,y)$ is the kernel of a propagator that satisfies
$$i \partial_t W(t,s) = K(t) W(t,s) , \quad W(s,s)=I$$
where $K(t;x,y)= \e^{i\int_{t_0}^t [V(r,x) - V(r,y)] dr} H(t;x,y) \left [ 1 - I(x,y) \right ].$
Thus 
\begin{equation}\label{eq:onestep}\abs{W(t,s;x,y)} \le I(x,y) + \int_s^t \sum_{\zeta\neq x} \abs{K(r,x,\zeta)} \abs{W(r,s;\zeta,y)} dr
\end{equation}

Note that  the exponential bound eq.\ \eqref{eq:expbound} holds with $K$ in place of $H$ with the same constants $m$ and $v$, i.e., we have
\begin{equation}\label{eq:expboundapp}v \ := \ \sup_t \sup_y \sum_{x\neq y} \e^{m|x-y|} \abs{K(t;x,y)} \ < \ \infty.\end{equation}
Because $\abs{U(t,x;x,y)}=\abs{W(t,s;x,y)}$, the desired estimate eq.~ \eqref{eq:fsop} essentially follows from eq.~ \eqref{eq:onestep} and Grönwall's lemma.  However, the argument is complicated a bit because we do not know \emph{a priori} that $\sum_x \e^{m|x-y|} U(t,s;x,y)$ is finite.  

Taking $m=0$ in eq.~\eqref{eq:expboundapp}, we see that $K(t)$ is a bounded map from $\ell^1(\Z^d) \rightarrow \ell^1(\Z^d)$ with norm no larger than $v$. Thus, 
$$\sum_{x} |W(t,s;x,y)|\ \le \ \e^{v|t-s|}. $$
It follows that 
$$\sum_{x}  \e^{m|x-y|_R} |W(t,s;x,y)| \ \le \ \e^{mR + v|t-s|}$$
where $|x-y|_R = \min(|x-y|,R)$.  Then from eqs.~\eqref{eq:onestep} and \eqref{eq:expboundapp} we conclude
\begin{align*}\sum_{x}  \e^{m|x-y|_R} |W(t,s;x,y)| \ &\le \ 1 + \int_s^t \sum_{\zeta} \left [ \sum_{x\neq\zeta }  \e^{m|x-\zeta|} |K(r,x,\zeta)| \right ] \e^{m|\zeta-y|_R} \abs{W(r,s;\zeta,y)} dr \\ 
&\le \ 1 + v \int_s^t \sum_{x} \e^{m|x-y|_R} \abs{W(r,s;x,y)} dr.
\end{align*}
Since both left and right hand sides are finite, Grönwall's lemma implies that
$$ \sum_{x} \e^{m|x-y|_R} \abs{W(t,s;x,y)}  \ \le \ \e^{v|t-s|}.$$
The finite speed of propagation bound eq.~\eqref{eq:fsop} now follows by taking $R \rightarrow \infty$. \qed 

\bibliographystyle{abbrv}

%\bibliography{library.bib}
%\printbibliography

\end{document}